\documentclass[letterpaper, 10 pt, conference]{ieeeconf}  

\IEEEoverridecommandlockouts                              

\usepackage{graphics} 
\usepackage{epsfig} 
\usepackage{times} 
\usepackage{amsmath} 
\usepackage{amssymb}  
\usepackage{fancyhdr}
\usepackage[dvipsnames]{xcolor}
\definecolor{antiquefuchsia}{rgb}{0.57, 0.36, 0.51}
\definecolor{DarkGreen}{rgb}{0.57, 0.8, 1}
\DeclareMathOperator{\diag}{diag}

\DeclareMathOperator{\rank}{rank}
\DeclareMathOperator{\Span}{span}
\DeclareMathOperator{\Null}{Null}
\DeclareMathOperator{\blkdiag}{blkdiag}
\newtheorem{thm}{Theorem}

\newtheorem{corollary}{Corollary}
\newtheorem{lem}{Lemma}
\newtheorem{assumption}{Assumption}
\newtheorem{remark}{Remark}
\newtheorem{definition}{Definition}
\newtheorem{problem}{Problem}
\title{\LARGE \bf
{Bearing-only formation control under persistence of excitation
}}

\author{Zhiqi~Tang$^{1,2}$, Rita Cunha$^{1}$, Tarek Hamel$^{2,3}$, Carlos Silvestre$^{1,4}$
\thanks{*This work was partially supported by the Project MYRG2015-00126-FST of the University
	of Macau; by the Macao Science and Technology, Development Fund under Grant FDCT/026/2017/A; by Funda\c{c}\~{a}o  para a Ci\^{e}ncia e a Tecnologia (FCT) through Project UID/EEA/50009/2019 and Project PTDC/EEI-AUT/5048/2014;
	and by the EQUIPEX project Robotex. The work of Z. Tang was supported by FCT through Ph.D. Fellowship PD/BD/114431/2016 under the FCT-IST NetSys Doctoral Program.}
\thanks{$^1$Institute for Systems and Robotics, Instituto Superior T\'{e}cnico, Universidade de Lisboa, Portugal.}
\thanks{$^2$I3S-CNRS, Universit\'{e} C\^{o}te d'Azur, Nice-Sophia Antipolis, France.}
\thanks{$^3$I3S-CNRS, Institut Universitaire de France, Nice-Sophia Antipolis, France.}
\thanks{$^4$Department of Electrical and Computer Engineering of the Faculty of Science and Technology of the University of Macau, Macao, China.}%
    \thanks{{\tt\small Email: zhiqitang@tecnico.ulisboa.pt, rita@isr.tecnico.ulisboa.pt, thamel@i3s.unice.fr, csilvestre@umac.mo}}
}

\begin{document}
\bstctlcite{IEEEexample:BSTcontrol}

\maketitle
\cfoot{\thepage}
\begin{abstract}
This paper addresses the problem of bearing-only formation control in $d~(d\geq 2)$-dimensional space by exploring persistence of excitation (PE) of the desired bearing reference.
By defining a desired formation that is bearing PE, distributed bearing-only control laws are proposed, which guarantee exponential stabilization of the desired formation only up to a translation vector.
The key outcome of this approach relies in exploiting the bearing PE to significantly
relax the conditions imposed on the graph topology to ensure
exponential stabilization, when compared to the bearing rigidity
conditions, and to remove the scale ambiguity introduced by
bearing vectors. Simulation results are provided to illustrate the performance of the proposed control method.

\end{abstract}

\section{INTRODUCTION}
The problem of formation control has been extensively
studied over the last decades both by the robotics and the
control communities. The main categories of solutions for formation control can be classified as \cite{oh2015survey}: 1)  position-based formation control \cite{ren2007distributed},  2) displacement-based formation control \cite{ren2005coordination}, 3) distance-based formation control \cite{anderson2007control} and more recently 4) bearing-based formation \cite{basiri2010distributed}. This latter category control has received growing attention due to its minimal requirements on the sensing ability of each agent. Early works on bearing-based formation control were mainly about controlling the subtended bearing angles that are measured in each agent's local coordinate frame  and were limited to the planar formations only \cite{basiri2010distributed,bishop2011very}. The main body of work however builds on concepts from bearing rigidity theory, which investigates the conditions for which a static formation is uniquely determined up to a translation and a scale given the corresponding constant bearing measurements.  Bearing rigidity theory in 2-dimensional space (also termed parallel rigidity) is explored in \cite{eren2003sensor,servatius1999constraining}. More recently it has been extended to arbitrary dimensions in \cite{zhao2016bearing} and a bearing-only formation control solution, that guarantees convergence to a desired formation that is centroid invariant and scale invariant with respect to the initial conditions of the formation,  is proposed.
By exploiting persistence of excitation (PE) of the
bearing vectors, we proposed in \cite{tang2020bearing,tang2020formationcontrol}:
 1) a relaxed bearing rigidity theory for leader-follower formations, which alleviates the constraints imposed on the graph topology required by the leader-first follower structure defined in \cite{trinh2019bearing}, and 2) bearing control laws achieving exponential stabilization of the leader-follower formation in terms of shape and scale, if the desired formation is bearing persistently exciting (BPE).

This paper provides a coherent generalization of our previous  solution to formations under general undirected graph topologies. Under the assumption that the desired formation is BPE, we propose control laws for a multi-agent system to track a desired formation using only bearing information.
 In particular, we show that under the bearing PE condition: 1) the exponential stabilization of the formation up to a translation  is achieved for any undirected graph that has a spanning tree (not necessarily bearing rigid) as shown in Fig.\ref{fig:pe_edges}-$(a1)$, $(b1)$ and $(b2)$); 2) scale ambiguity, which is a characteristic of bearing rigidity, can be removed without the need to measure the distance between any two agents. The main focus of the paper is pointing out general and explicit
PE conditions whose satisfaction ensures exponential stabilisation of the formation to the desired one in terms of shape and scale.

The body of the paper is organized as follows. Section II presents mathematical background on graph theory and formation control. Section III describes several properties related to bearing persistence of excitation in arbitrary dimensional spaces. Section IV presents a bearing-only formation control law along with stability analysis. Section V shows the performance of the proposed control strategy in two different scenarios. The paper concludes with some final comments in Section VI.
\begin{figure}[!htb]
	\centering
	\includegraphics[width=2.1in]{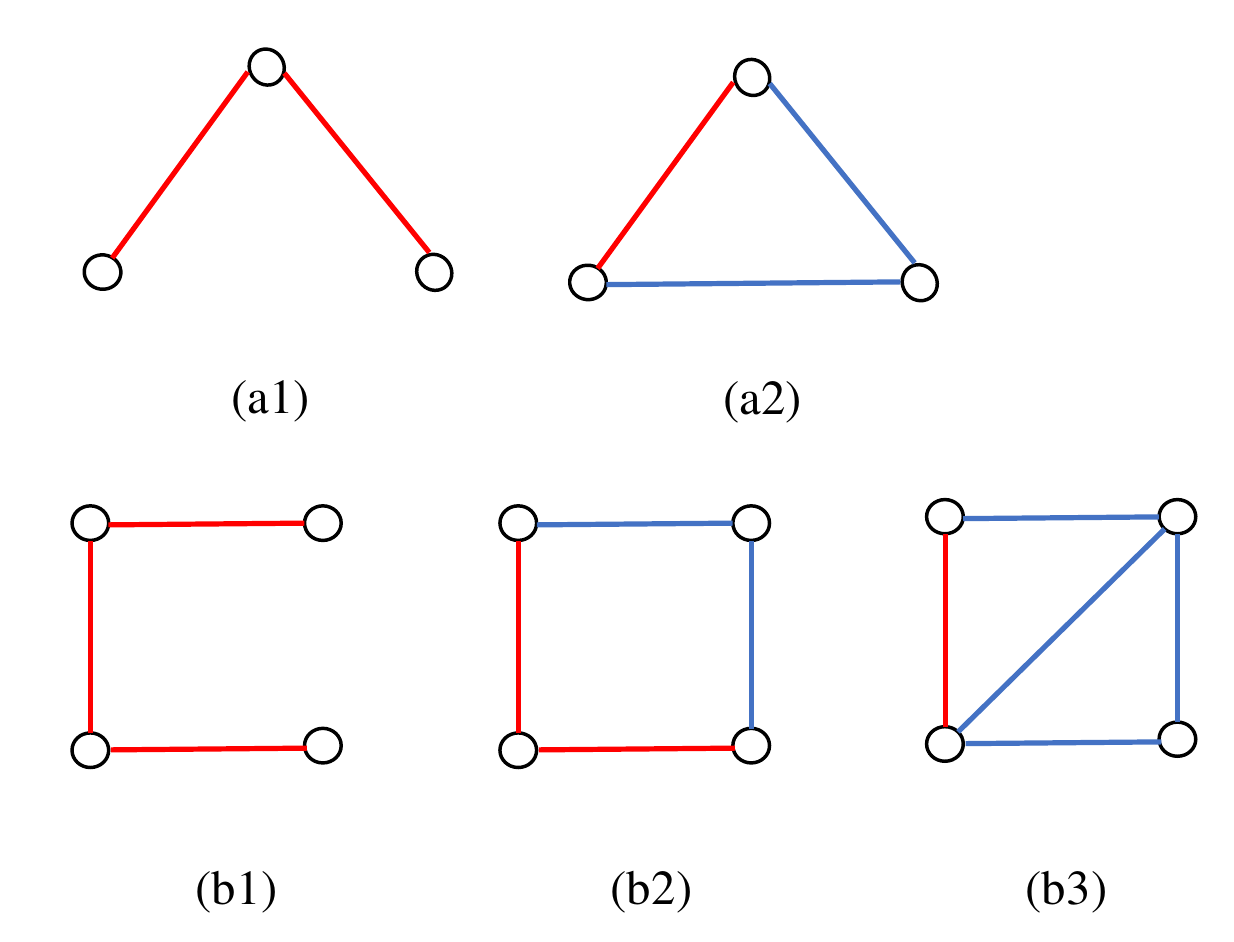}	
	\caption{Examples of bearing persistently exciting formations. Red lines represent edges for which the corresponding bearing vector are persistently exciting and blue lines represent edges for which the corresponding bearing vectors are not necessarily persistently exciting.}
	\label{fig:pe_edges}
\end{figure}
\section{Preliminaries on graph theory and formation control}
\subsection{Notation}
Let $\mathbb{S}^{d-1}:=\{y\in\mathbb{R}^d:\|y\|=1\}$ denote the $(d-1)$-Sphere ($d \geq 2$) and $\|.\|$ the euclidean norm. The null space and rank of a matrix are denoted by $\Null(.)$ and $\rank(.)$, respectively, and $\lambda_{\max}(.)(\lambda_{\min}(.))$ represents the maximum (minimum) eigenvalue of its matrix argument. The matrix $I_d$ represents the identity matrix of dimension $d\times d$. The operator $\otimes$ denotes the Kronecker product, $1_n = [1,\ldots,1]^T \in \mathbb R^{n}$ denotes the column vector of ones, and $\diag(A_i) = \blkdiag \{A_1, \ldots, A_n\} \in \mathbb{R}^{nd\times nd}$ denotes the block diagonal matrix with elements given by $A_i\in \mathbb{R}^{d\times d}$ for $i=1,\ldots,n$.
For any $y\in \mathbb{S}^{d-1}$, we can define the projection operator $\pi_y$
\begin{align}
\pi_y := I_d - y y^{\top} \geq 0, \label{eq:pi}
\end{align}
as the orthogonal projection operator in $\mathbb{R}^d$ onto the $(d-1)$-dimensional $(d\geq 2)$ vector subspace orthogonal to $y$.
\subsection{Graph theory}\label{sec:LF}
Consider a system of $n~(n\geq 2)$  connected agents. The underlying interaction topology can be modeled as an undirected graph $\mathcal{G} := (\mathcal{V}, \mathcal{E})$, where $\mathcal{V}=\{1,\ldots,n\}$ is the set of vertices and $\mathcal{E} \subseteq \mathcal{V} \times \mathcal{V}$ is the set of undirected edges. Two vertices $i$ and $j$ are called adjacent (or neighbors)
when $\{i,j\}\in \mathcal{E}$. The set of neighbors of agent $i$ is denoted by $\mathcal{N}_i:=\{j\in\mathcal{V}|\{i,j\}\in\mathcal{E}\}$. If $j\in\mathcal N_i$, it follows that $i\in \mathcal N_j$, since the edge set in an undirected graph consists of unordered vertex pairs. Define $m=|\mathcal E|$, where $|.|$ denotes the cardinality of a set. A graph $\mathcal G$ is connected if there exists a path between every pair of vertices in $\mathcal G$ and in that case $m \geq n-1$. A graph $\mathcal G$ is said to be acyclic if it has no circuits. A tree is a connected acyclic graph.
A spanning tree of a graph $\mathcal G$ is a tree of $\mathcal G$ having all the vertices of $\mathcal G$.
An orientation of an undirected graph is the assignment of a direction to each edge. An oriented graph is an undirected graph together with an orientation. The incidence matrix $H\in \mathbb{R}^{m\times n}$ of an oriented graph is the $\{0,\pm 1\}$-matrix with rows indexed by edges and columns by vertices: $[H]_{ki}=1$ if vertex $i$ is the head of the edge $k$, $[H]_{ki}=-1$ if it is the tail, and $[H]_{ki}=0$ otherwise. For a connected graph, one always has $H1_n=0$ and rank$(H)=n-1$.
\subsection{Formation control} \label{subsection:formation control}
Consider an undirected graph $\mathcal{G}=(\mathcal{V},\mathcal{E})$, let $p_i  \in \mathbb{R}^d$ denote the position of each agent $i\in \mathcal V$ expressed in an inertial frame common to all agents and $v_i  \in \mathbb{R}^d$ denote the velocity control input, such that $\dot p_i=v_i$. Then, the stacked vector $\boldmath{p}=[p_1^\top,...,p_n^\top]^\top\in \mathbb{R}^{dn}$ is called a configuration of $\mathcal{G}$. The graph $\mathcal{G}$ and the configuration $\boldmath{p}$ together define a formation $\mathcal{G}(p)$ in the d-dimensional space. Let $v:=\dot p=[v_1^\top,\ldots,v_n^\top]^\top\in \mathbb R ^{dn}$ denote the stacked vector of velocity control inputs. For a formation, define the relative position
\begin{equation}\label{eq:pij}
p_{ij}:=p_j-p_i,\ \{i,j\}\in\mathcal E,
\end{equation}
as long as $\|p_{ij}\|\ne 0$, the bearing of agent $j$ relative to agent $i$ is given by the unit vector
\begin{equation}
\label{eq:gij}
g_{ij} := p_{ij}/\|p_{ij}\| \in \mathbb{S}^{d-1}.
\end{equation}
Consider an arbitrary orientation of the graph and denote
$$e_k:= p_{ij},\ k\in \{1,\ldots,m\},$$
as the edge vector with assigned direction such that $i$ and $j$ are, respectively, the initial and the
terminal nodes of $e_k$. Denote the corresponding bearing vector by
$$g_k:= \frac{e_k}{\|e_k\|}\in \mathbb{S}^{d-1},\ k\in \{1,\ldots,m\}.$$
Define the stacked vector of edge vectors $e=[e_1^\top,...,e_m^\top]^\top=\bar{H}p$, where $\bar H=H \otimes I_d$. 
\subsubsection{Formation control using relative position measurements}\label{formation control}
In this problem setup, the agents sense relative positions of their neighbors. The formation control objective is to derive the configuration $p$ to the desired one up to translation, i.e. \cite{mesbahi2010graph,oh2015survey}.
The graph Laplacian matrix is introduced as
\begin{equation} \label{eq:Laplacian}
L=\bar H^\top \bar H
\end{equation}
Note that if the graph is connected, or equivalently has a spanning tree, $\rank(L)=dn-d$, $\Null(L)=\Span\{U\}$, with $U=1_n\otimes I_d$. Let $\lambda_i$ denote the $i$th eigenvalue of $L$ under a non-increasing order and note that $\lambda_{dn-d}$ is the smallest positive eigenvalue of $L$.
\subsubsection{Formation control using bearing measurements}
In this setting, the agents measure the relative directions to their neighbors (bearings) and the objective of the formation control is to drive the configuration $p$ to the desired configuration up to a translational and a scaling factor, i.e. \cite{zhao2016bearing,zhao2019bearing}. The bearing Laplacian matrix is introduced as
\begin{equation}\label{eq:bearingL}
L_B=\bar H^\top \Pi\bar H
\end{equation}
where $\Pi=\diag(\pi_{g_k})$. 
Since $\Span\{U,p\}\subset \Null(L_B)$ it follows that $\rank(L_B)\le dn-d-1$. According to \cite{zhao2016bearing} (in which only constant bearing are considered), if the formation is infinitesimally bearing rigid then $\rank(L_B)= dn-d-1$ and $\Null(L_B)=\Span\{U,p\}$.
\section {Bearing persistence of excitation in $\mathbb R^d$}
In this section, we describe the concept of  persistence of excitation for bearing vectors and characterize BPE formations (i.e. formations that can be uniquely determined up to a translational factor using only bearing and velocity control inputs information).
\subsection{Persistence of excitation on bearings}
\begin{definition}
	A positive semi-definite matrix $\Sigma(t)\in \mathbb{R}^{n\times n}$, is called \textit{persistently exciting} (PE) if there exists $T>0$ and $\mu>0$ such that for all $t$
	\begin{equation}
	\int_{t}^{t+T}\Sigma(\tau)d\tau\ge\mu I. \label{eq:pe}
	\end{equation}
	\label{def:pe of matrix}
\end{definition}
\begin{definition}\label{def:pe}
	A direction $y(t)\in \mathbb{S}^{d-1}$, is called \textit{persistently exciting} (PE) if the matrix $\pi_{y(t)}$ satisfies the PE condition from Definition \ref{def:pe of matrix}.
\end{definition}
\begin{lem}\label{lem:pe_norm}
	For a direction $y(t)\in \mathbb{S}^{d-1}$, assume that $\dot y(t)$ is uniformly continuous, then $y(t)$ is PE ($\pi_{y(t)}$ satisfies \eqref{eq:pe}) if and only if:
	
	There exists $(T,\epsilon)>0$ such that $\|\dot y(\tau)\|\ge \epsilon$, $\forall \tau \in [t,\ t+T]$.
	
\end{lem}
\begin{proof}
	The proof of this lemma is given in \cite[Appendix 6.1]{le2017observers}.
\end{proof}

\subsection{BPE formation and relaxed bearing rigidity}
Now we introduce a relaxed persistence of excitation condition specifically developed to characterize the bearing Laplacian matrix.
\begin{definition}
	Consider the Laplacian $L$ and the bearing Laplacian $L_B$ defined in \eqref{eq:Laplacian} and \eqref{eq:bearingL}, respectively. The bearing Laplacian matrix is called \textit{persistently exciting} (PE) if  there exists $T>0$ and $\mu>0$  such that for all $t$
	\begin{equation}
	\int_{t}^{t+T}L_B(\tau)d\tau\ge\mu L. \label{eq:pe_LB}
	\end{equation}
	\label{def:pe of LB}
\end{definition}
\begin{remark}
	One can verify that the PE condition for the bearing Laplacian introduced in Definition \ref{def:pe of LB} is less restrictive than the PE condition from Definition \ref{def:pe of matrix}. In particular, having a matrix $\Pi$ that is PE is a sufficient but not necessary condition to ensure that $L_B = \bar{H}^\top\Pi\bar{H}$ is also PE.
\end{remark}
\begin{definition}\label{bpe}
	A formation $\mathcal G(p(t))$ is \textit{bearing persistently exciting} (BPE) if $\mathcal G$ has a spanning tree and its bearing Laplacian matrix is PE.
\end{definition}
\begin{thm}\label{thm:shape}
	Consider a formation $\mathcal G (p(t))$ defined in $\mathbb R ^{d}$ along with bearing measurements $\{g_k\}_{ k\in\{1\ldots m\}}$ of an arbitrary orientation of the graph. Assume that the velocity control inputs $\{v_i\}_{ i\in \{1 \ldots n\}}$ are bounded and known. If the formation $\mathcal G (p(t))$ is BPE then the configuration $p(t)$ can be recovered up to a translational  vector in $\mathbb R ^{d}$.
\end{thm}
\begin{proof}
	Consider the stacked velocity vector $v(t)=[v_1^{\top}(t),...,v_n^{\top}(t)]^\top\in \mathbb{R}^{dn}$ and  let ${\hat p}$ denote the estimate of $ p$ with dynamics:
	\begin{equation} \label{eq:observer}\small
	\begin{aligned}
	\dot{\hat p}=v-L_B(t)\hat p,
	\end{aligned}
	\end{equation}
	with arbitrary initial conditions.
	Consider the error variable $\zeta=\hat p(t)-p(t)-\frac{UU^\top(\hat p(0)-p(0))}{n}$ (recall that $U=1_n\otimes I_d$), the corresponding dynamics can be obtained from \eqref{eq:observer}:
	\begin{equation} \label{eq:observer_dynamics}
	\dot{\zeta}=-L_B(t) \zeta.
	\end{equation}
	Due to the fact that $U^\top (\dot{\hat p}(t) -\dot p(t))\equiv 0$, one has $U^\top\zeta(t)\equiv0$. Since the formation is BPE, one has $\forall x\in S^{nd-1}$ satisfying $U^\top x\equiv0$, there exists a $T>0$ and $\mu >0$ such that, $\forall t\ge 0$, $x^\top \int_t^{t+T}L_B(\tau)d \tau x\ge \mu x^\top \bar H^\top \bar Hx\ge \mu \lambda_{dn-d}$, where $\lambda_{dn-d}>0$ is the smallest positive eigenvalue of $\bar H^\top \bar H$ (see Sect. \ref{formation control}). Using a similar  arguments as the proof of \cite[Lemma 5]{loria2002uniform}, one can conclude that the equilibrium $\zeta=0$ is uniformly globally exponentially (UGE) stable. Then we can conclude that $\hat{p}$ converges UGE to the unique $p$ up to a translational vector $\frac{UU^\top(\hat p(0)-p(0))}{n}$.
\end{proof}
\begin{remark}
	Note that for a BPE formation, the shape and the size of the formation may be time-varying. This includes similarity transformations (a combination of rigid transformation and scaling) involving a time-varying rotation. In this case, it is straightforward to show that for any bearing formation the bearing measurements are invariant to translation and scaling but change with rotation such that $g_{ij}(t)=R(t)^\top g_{ij}(0)\forall (i,j)\in\mathcal{E}$ (with $R(t) \in SO(3)$ the rotation part of the similarity transformation).
	This implies that there exists a similarity transformation in which $R(t)$ is time-varying such that the formation $\mathcal G(p(t))$ is BPE.
\end{remark}
\begin{definition}
	A formation $\mathcal G(p(t))$ is called \textit{relaxed bearing rigid}, if it is bearing PE and subjected to a similarity transformation.
\end{definition}
\begin{corollary}
	If the formation is relaxed bearing rigid, then the result of Theorem \ref{thm:shape} applies.
\end{corollary}
\begin{proof}
	The proof is analogous to the proof of Theorem \ref{thm:shape}. It is omitted here for the sake of brevity.
\end{proof}

\subsection{Properties of BPE formations}
We explore here the relationship between the number of PE bearings in a formation and guaranteeing that a formation is BPE. More precisely, we focus on the specific bearing vectors in the formations that have to be PE in order to guarantee that a formation is BPE .
\begin{lem} \label{lem:rigid}
	Consider a formation $G(p(t))$ defined in $ \mathbb R^{d}$, If $\forall t\ge 0$, $\rank(L_B(t))=dn-d-1$, then $\Null(L_B(t))=\Span\{U,\ p(t)\}$.
	
\end{lem}
\begin{proof}
	The proof is same as stated in \cite[Theorem 4]{zhao2016bearing} since \cite[Theorem 4]{zhao2016bearing} is still valid for time-varying cases.
\end{proof}

\begin{lem} \label{lem:rigid to persistent}
	For a formation $\mathcal G(p(t))$ defined in $\mathbb R^d$, assume $\rank(L_B(t))=dn-d-1, \ \forall t\ge 0$. Then $\mathcal G(p(t))$ is bearing persistently exciting if and only if at least one bearing $g_k(t),\ k\in\{1,\ldots,m\}$ is PE.
\end{lem}
\begin{proof}
	Since $\rank(L_B(t))=dn-d-1, \ \forall t \ge 0$, $\mathcal G$ has a spanning tree and $\Null(\Pi(t))=\Span\{\bar H p(t)\}$ (from Lemma \ref{lem:rigid}). In order to prove that the formation is BPE, it suffices to show that its bearing Laplacian matrix is PE. Let $S=\{\mathring p \in S| \mathring p=[\mathring p_1^\top,\ldots,\mathring p_n^\top]\in\mathbb R^{dn}\}$ be the set of all possible fixed configurations under the formation $\mathcal G(\mathring p)$ leading to $\rank(L_B)=dn-d-1$. 
	This in turn implies that there exists a $w=[w_1^\top,\ldots, w_k^\top,\ldots,w_m^\top]^\top=\bar H\mathring p$ and a positive constant $\epsilon$ such that $\|w_k\|=\|\mathring p_i-\mathring p_j\|\ge \epsilon, \ \forall k \in \{1,\ldots, m\}$ (i.e. $\mathring g_k=\frac{w_k}{\|w_k\|} $ is well define $\forall k \in \{1,\ldots, m\}$).
	
	Now to prove the 'if' part of the lemma we use the fact that there exists at least one bearing vector $g_q,\ q\in\{1,\ldots,m\}$ is PE. This implies that there exist two constant $T>0$, $\mu_q>0$ such that $\forall t\ge 0$ and for all fixed $\mathring p\in S$ leading to $w=\bar H \mathring p$, we have
	\begin{equation}\small
	\begin{aligned}
		w^\top \int_t^{t+T}\Pi(\tau)d\tau w=\sum_{k=1}^m c_k w_k^\top \int_t^{t+T}\pi_{g_k(\tau)}d\tau w_k \ge c_q\mu_q\|w_q\|^2.
	\end{aligned}
	\end{equation}
By choosing $0<\mu <c_q\mu_q\frac{\|w_q\|^2}{\|w\|^2}$, one gets $\mathring p ^\top \bar H^\top \int_t^{t+T}\Pi(\tau)d\tau \bar H\mathring p\ge \mu \mathring p ^\top \bar H^\top\bar H\mathring p$ which implies that $L_B(t)$ is PE.

Now to prove the 'only if', we will proceed by contradiction. Assume that none of the bearing vector is PE which implies that $\forall \mu_k>0$, $\forall T>0$, $\exists t$ and $\exists w=\bar H \mathring p$, such that $w_k^\top \int_t^{t+T}\pi_{g_k(\tau)}d\tau w_k< \mu_k\|w_k\|^2,\ \forall k\in\{1,\ldots,m\}$. Since $L_B(t)$ is PE, there exists $T>0$ and $\mu>0$ such that, $\forall t$ and $\forall w=\bar H \mathring p$, $w^\top \int_t^{t+T}\Pi(\tau)d\tau w\ge \mu\|w\|^2$. Choose $c_k\mu_k\le \frac{\mu\|w\|^2}{m\|w_k\|^2}$, we can conclude that, $\exists t$ and $\exists w=\bar H \mathring p$
\begin{equation}\small
w^\top \int_t^{t+T}\Pi(\tau)d\tau w=\sum_{k=1}^m c_kw_k^\top \int_t^{t+T}\pi_{g_k(\tau)}d\tau w_k < \mu\|w\|^2
\end{equation}
which yields a contradiction.
\end{proof}
\begin{lem}\label{lem:min_pe}
	Consider a formation $\mathcal{G}(p(t))$ defined in $\mathbb R^{d}$. Assume $\mathcal G$ is acyclic and has a spanning tree, then the formation is BPE if and only if $g_k(t)$ satisfies the PE condition for all $k\in\{1,\ldots,m\}$.
\end{lem}
\begin{proof}
	Since $\mathcal G$ is acyclic and has a spanning tree, $m=n-1$. According to Definition \ref{bpe}, the proof of the lemma is equivalent to showing that $L_B(t)$ is PE if and only if the corresponding bearing vectors $g_k(t),\ \forall k=\{1,\ldots,n-1\}$ satisfies the PE condition.
	
	If $g_k(t)$ satisfies the PE condition  $\forall k=\{1,\ldots,m\}$, this implies that the matrix $\Pi(t)$ is PE and hence it is obvious to conclude that $L_B(t)$ is PE.
	Conversely, if $L_B(t)$ is PE then there exist $T>0$ and $\mu>0$ such that, $\forall t \ge 0$,  $\int_{t}^{t+T}L_B(\tau)d\tau\ge\mu L$. Now, since the $\bar{H}$ is a constant matrix with $\rank(\bar{H})=d(n-1)$ and $\Pi(t) \in \mathbb{R}^{d(n-1)\times d(n-1)}$ it follows that $\Pi(t) \in \mathbb{R}^{d(n-1)\times d(n-1)}$ should satisfy the PE condition in equation \eqref{eq:pe}.  This in turn implies that each $g_k(t)$ satisfies the PE condition in Definition \ref{def:pe},  $\forall k\in\{1,\ldots,n-1\}$.
	
	
\end{proof}

\begin{lem}
	Consider a formation $\mathcal{G}(p(t))$ defined in $\mathbb{R}^2$ along with bearing measurements $\{g_k\}_{ k\in\{1\ldots m\}}$ of an arbitrary orientation of the graph. If the formation is BPE, then
	the number of PE bearing vectors, $\bar m$, satisfies the condition:
	\begin{enumerate}
		\item $\bar m\ge 1,\ \text{if} \ m\ge 2n-3$,\\
		\item $\bar m\ge j+1,\ \text{if} \ m=2n-3-j\ ( j\in\{1,\ldots,n-2\}).$\\
	\end{enumerate}
\end{lem}
\begin{proof}
	The proof of item 1) is similar to the proof of the 'only if' part in Lemma \ref{lem:rigid to persistent}. It has been omitted here to save space. Now, in order to show that item 2) is valid,  we have to verify that inequality \eqref{eq:pe_LB} is satisfied when $\bar m\ge j+1,$ in the case of $m=2n-3-j\ ( j\in\{1,\ldots,n-2\})$. That is there exists $  \mu>0$ and $ T>0$, $\forall t \ge 0$ and $\forall x\in R^{2n}$ such that $\bar{H}x \neq 0$, we have $x^\top \int_t^{t+ T}L_B(\tau)d\tau x\ge  \mu x^\top Lx$ or equivalently $w^\top\int_t^{t+ T}\Pi(\tau)d\tau w\ge \mu \|w\|^2$, with $w=\bar{H} x\in \mathbb R^{2m}$.
	
	We proceed by contradiction. Assume that $\bar m\le j$.
	Since we have $m-\bar{m}\geq 2n-2j-3$ non-PE bearings and for each non-PE bearing $g_k$
	there is a
	 $\lambda_{\min}(\int_{t}^{t+T}\pi_{g_k(\tau)}d\tau)<\mu$, it is straightforward to verify that $\lambda_{2n-2}(\int_t^{t+T}\Pi(\tau)d\tau)< \mu$ ($\lambda_i(.)$ represents the $i$th eigenvalue of a symmetric matrix under a non-increasing order).
	
	Now, using the fact that $\rank(\bar H)=2n-2$, we can ensure that if $x=(x_1^\top, \ldots,x_n^\top)$ has $2n$ independent entries (each $x_i \in \mathbb{R}^2$), then there exists a $w=\bar H x$ with $2n-2$ independent entries such that $w^\top\int_t^{t+T}\Pi(\tau)d\tau w <\mu \|w\|^2$, which yields a contradiction.
\end{proof}

Figure \ref{fig:pe_edges} illustrates that when $m\le 2n-3$, the minimal number of PE bearing vectors decreases as the edges number $m$ increases.


\section{Bearing-only formation control}
In this section we will propose a bearing-only formation control law for a multi-agent system provided the desired formation is BPE.

Consider the formation $\mathcal{G}(p)$ defined in Section \ref{subsection:formation control}, where each agent $i\in\mathcal{V}$ is modeled as a single integrator with the following dynamics:
\begin{equation} \label{eq:double integrator}
\dot{p}_i=v_i
\end{equation}
where $v_i\in\mathbb{R}^d$ is the velocity control input. Similarly, let $p_i^*(t)$ and $v_i^*(t)\in\mathbb{R}^d$ denote the desired position and velocity of the $i$th agent, respectively, and define the desired relative position vectors $p_{ij}^*$ and bearings $g_{ij}^*$, according to \eqref{eq:pij} and \eqref{eq:gij}, respectively. Let $\boldmath{p}^*(t)=[p_1^{*\top}(t),...,p_n^{*\top}(t)]^\top\in \mathbb{R}^{dn}$ be the desired configuration. Let $\{e_{k}^*(t)\}_{k\in\{1,...,m\}}$ and  $\{g_{k}^*(t)\}_{k\in\{1,...,m\}}$ be the set of all desired edge vectors and desired bearing vectors, respectively, under an arbitrary orientation of the graph.

We assume that the $n$-agent system satisfies the following assumptions.

\begin{assumption} \label{ass:construction}
	The sensing topology of the group is described by a undirected graph $\mathcal{G}(\mathcal{V},\mathcal{E})$ which has a spanning tree. Each agent $i\in \mathcal V$ can measure the relative bearing vectors $g_{ij}$ to its neighbors $ j\in \mathcal{N}_i$.
\end{assumption}
\begin{assumption}\label{ass:desired}
	The desired velocities $v_i^*(t)$ ($i\in \mathcal V$) are bounded and known, the resulting desired bearings $g_{ij}^*(t)$ are well-defined and the desired formation $\mathcal G(p^*(t))$ is BPE, for all $t\ge 0$ .
\end{assumption}
\begin{assumption}\label{ass:collision}
As the formation evolves in time, no inter-agent collisions and occlusions occur. In particular, we assume that there exists a positive number $\epsilon$ such that $\forall t \ge 0,\ \|p_{ij}(t)\|>\epsilon,\ \{i,j\}\in\mathcal{E}$.\end{assumption}

With all these ingredients, we can define the bearing-only formation control problem as follows.
\begin{problem}
	Consider the system \eqref{eq:double integrator} and the underlying formation $\mathcal{G}(p)$. Under Assumptions \ref{ass:construction}-\ref{ass:collision}, design distributed control laws based on bearing measurements that guarantee exponential stabilization of the actual formation to the desired one up to a translational vector.
\end{problem}

\subsection{A bearing-only control law}

For each agents $i\in\mathcal{V}$, define the position errors $\tilde{p}_{i}:=p_{i}-p_{i}^*$ along with the following kinematics:
\begin{equation} \label{eq:states_f}
\dot{\tilde{p}}_{i}=v_i-v_i^*
\end{equation}
Consider the following control law for each agent $i\in\mathcal{V}$
\begin{equation}
v_i=-k_p\sum_{j\in \mathcal{N}_i}\pi_{g_{ij}}p_{ij}^*
+v_i^*. \label{eq:ui}
\end{equation}
where $k_p$ is a positive gain. Let $\tilde{p}:=p-p^*$ be the configuration error. Using the control law \eqref{eq:ui} for $i\in\mathcal V$, one gets:
\begin{equation}\label{tildep}
\dot{ \tilde{p}}=-k_pL_B(t)\tilde{p}.
\end{equation}

\begin{lem} \label{lem:invariant}
Consider the configuration error $\tilde{p}$ governed by \eqref{tildep}, then the relative centroid vector $\bar p:=\frac{1}{n}U^\top \tilde p$ is invariant.
\end{lem}
\begin{proof}
	The derivative of the relative centroid vector is
	\begin{equation}
	\dot {\bar p} =U^\top\dot {\tilde{p}}/n=-\frac{k_p}{n}U^\top L_B(t)\tilde{p}\equiv 0, \label{eq:dotscale}
	\end{equation}
	due to the fact that $\Span\{U\}\subset\Null(L_B(t))$. Thus the relative centroid $\bar p(t)$ is invariant respect to $\bar p(0)$.
	\end{proof}

\subsection{Exponential stabilization of the formations}
\begin{thm} \label{lem:2agent}
	Consider the error dynamics \eqref{eq:states_f} along with the control law \eqref{eq:ui}.
	If the Assumptions \ref{ass:construction}-\ref{ass:collision} are satisfied, then the equilibrium point $\tilde{p}=\frac{1}{n}UU^\top \tilde{p}(0)$ is exponentially stable (ES).
\end{thm}

\begin{proof}
	Define a new variable $\delta:=\tilde{p}-\frac{1}{n}UU^\top \tilde{p}(0)$ and
a Lyapunov function:
\begin{equation}
\begin{aligned}
\mathcal{L}&=\frac 1 2 \|\delta\|^2,
\end{aligned}\label{eq:L}
\end{equation}
with the following time derivative negative semi-definite ($L_B(p(t))\geq 0$)
\begin{equation}
\begin{aligned}
\dot{\mathcal{ L}} &=-k_p\delta^\top L_B(p(t))\delta.
\end{aligned}\label{eq:dotL_semi}
\end{equation}
Thus we can conclude that $\delta(t)$ is bounded. Equation \eqref{eq:dotL_semi} can be represented as
\begin{equation}\small
\begin{aligned}
\dot{\mathcal{ L}} &=-k_p\tilde{p}^\top L_B(p(t))\tilde{p}=-k_p\sum_{k=1}^m e_k^{*\top}\pi_{g_{k}}e_k^*\\
&=-k_p\sum_{k=1}^m\frac{\|e_k^*\|^2}{\|e_k\|^2}\tilde{e}_k^{\top}\pi_{g_{k}^*}\tilde{e}_k\le-\gamma\delta^\top L_B(p^*(t))\delta
\end{aligned}\label{eq:dotLre}
\end{equation}
where $\tilde{e}_k:=e_k-e_k^*$ and $\gamma= k_p\frac{\min_{k=1,\dots,m}\|e_k^*(t)\|^2}{\|\bar H\|^2(\|\delta(0)\|+\|U^\top \tilde{p}(0)\|+\max\|p^*(t)\|)^2}>0$.
From now on, we use a similar argument as shown in proof of \cite[Lemma 5]{loria2002uniform}. Take the integral of \eqref{eq:dotLre}, we have
\begin{equation}
\mathcal{L}(t+T)-\mathcal{L}(t)\le-\gamma\int_{t}^{t+T}\|L_B(p^*(\tau))^{\frac 1 2}
\delta(\tau)\|^2 d \tau \label{eq:dotL_int}
\end{equation}
where the solution
\begin{equation}
\delta(\tau)=\delta(t)-\int_{t}^{\tau}L_B(p(s))\delta(s)ds. \label{eq:dote_int}
\end{equation}
	 Substituting \eqref{eq:dote_int} in $\|L_B(p^*(\tau))^{\frac 1 2}\delta(\tau)\|^2$ and use $\|L_B(p^*(\tau))^{\frac 1 2}\|\le \|\bar H\|$, $\|a+b\|^2\ge[\rho/(1+\rho)]\|a\|^2-\rho \|b\|^2$ and Schwartz inequality, we obtain
\begin{equation}\small
\begin{aligned}
&\mathcal{L}(t+T)-\mathcal{L}(t)\le
-\frac{\rho\gamma}{1+\rho}\delta(t)^\top \int_t^{t+T}L_B(p^*(\tau))d\tau\delta(t)\\
&+\gamma\rho  T \|\bar H\|
\int_t^{t+T}\int_t^{\tau} \|L_B(p(s))\delta(s)\|^2ds d\tau.
\end{aligned}
\end{equation}
Using the conclusion of Lemma \ref{lem:invariant} together with the assumption that the desired formation is BPE, we can conclude that $\delta(t)^\top\int_t^{t+T}L_B(p^*(\tau))d\tau\delta(t)\ge\mu \delta(t)^\top\bar H^\top \bar H \delta(t)\ge\mu \lambda_{dn-d} \|\delta(t)\|^2$ (recall that $\lambda_{dn-d}$ is the smallest positive eigenvalue of $\bar H ^\top \bar H$). Hence, we have
\begin{equation}\small
\begin{aligned}
&\mathcal{L}(t+T)-\mathcal{L}(t)\le-\gamma\mu\lambda_{dn-d}\rho/(1+\rho)\|\delta(t)\|^2\\
&+\gamma\rho T \|\bar H\|^3 \int_{t}^{t+T}\int_{t}^{\tau}\|L_B(p(s))^{\frac 1 2}\delta(s)\|^2 ds d\tau.
\end{aligned}
\end{equation}
Using the similar argument as shown in \cite[Lemma 5]{loria2002uniform}, one gets
\begin{equation*}
\frac{\rho\gamma\mu\lambda_{dn-d}}{1+\rho}\|\delta(t)\|^2\le(1+\gamma\rho T^2\|\bar H\|^3)[\mathcal{L}(t)-\mathcal{L}(t+T)].
\end{equation*}
Which in turn implies that for $\sigma:=\frac{2\rho\gamma\mu\lambda_{dn-d}}{(1+\rho)(1+\gamma\rho T^2\|\bar H\|^3)}$ and any $\rho$ such that $0<\sigma<1$, one has:
\begin{equation}
\mathcal{L}(t+T)-\mathcal{L}(t)\le-\sigma\mathcal{L}(t),\\
\end{equation}
and  hence $\delta=0$ is exponentially stable.\\
\end{proof}
\begin{remark}
	Assumption \ref{ass:collision} requested in the above Theorem relies on the evolution of state variables. It serves to show that if there is no collision or occlusion, the bearings are well-defined and the proposed control design yields the desired convergence properties. Trying to more specifically characterize the set of all initial conditions for which the system's solutions avoid collision and occlusion is out of the scope of the present paper.
\end{remark}

\section{Simulation Results}\label{sec:sim}
In this section, simulation results are provided to illustrate the effectiveness of the proposed control law for a four-agent system in both 2-D and 3-D space. In 2-D space, we consider a relaxed bearing rigid desired formation under the graph topology which has a single spanning tree. The desired four agents form a squared shape in $\mathbb{R}^2$ that rotates about its center in the meanwhile translating along $x$-axis (Fig.~\ref{fig:2D}) such that,
the desired positions are  $p_i^* (t) = R(t)^\top p_i^*(0)+[t/10 \ 0]^\top$, with \scalebox{0.9}{$R(t)=\begin{bmatrix}
	\cos(\frac {\pi} {6}t)& -\sin(\frac {\pi} {6}t) \\ \sin(\frac  {\pi} {6}t) &\cos(\frac  {\pi} {6}t)
	\end{bmatrix}$}, $p_1^*(0)=[0\ 1]^\top,p_2^*(0)=[1\ 0 ]^\top, p_3^*(0)=[0\ -1]^\top$ and $p_4^*(0)=[-1 \ 0]^\top$
The initial conditions are chosen such that $U^\top \tilde{p}(0)=0$ (the initial centroid coincides with the initial centroid of the desired formation): $p_1(0)=[1\ 1]^\top$, $p_2(0)=[-1\ 2 ]^\top$, $p_3(0)=[1\ -1 ]^\top$, and $p_4(0)=[-1\ -2 ]^\top$, which implies that the convergence of $\delta=0$ ensures the convergence of $\tilde{p}$ to $0$ (by Lemma \ref{lem:invariant}). The chosen gain is $k_p=1$. Fig. \ref{fig:2D} depicts the three snapshots of the agents during time evolution of the formation and it shows that the four agents converge to the desired trajectories. Fig. \ref{fig:error1} shows the time evolution  and the convergence of $\|\tilde{p}(t)\|$ to $0$ and, hence we can conclude that the formation achieves to the desired shape and scale without need of bearing rigidity.

In 3-D space, we consider a desired formation such that the four agents form a pyramid shape in $\mathbb{R}^3$, that rotates about agent 1. The graph is such that $N_i=\{j\in\mathcal{V}: j\ne i\}$. The desired position of the agents are
 $p_i^* (t) = R(t)^\top p_i^*(0)$, with \scalebox{0.8}{$R(t)=\begin{bmatrix}
	\cos(\frac {\pi} {6}t)& -\sin(\frac {\pi} {6}t) &0\\ \sin(\frac  {\pi} {6}t) &\cos(\frac  {\pi} {6}t)& 0\\1 &0& 0
	\end{bmatrix}$}, $p_1^*(0)=[0\ 0\ 0]^\top,p_2^*(0)=[1\ 0\ 0 ]^\top, p_3^*(0)=[0.5\ -\sqrt{3}/2\ 0]^\top$ and $p_4^*(0)=[\sqrt{3}/2 \ -0.5 \ 1]^\top$.
 The initial conditions are $p_1(0)=[1\ 1 \ 0]^\top$, $p_2(0)=[-1\ 2 \ 1]^\top$, $p_3(0)=[-2\ 0 \ -1]^\top$, and $p_4(0)=[-1\ 2 \ 2]^\top$. The gain is $k_p=0.5$. Fig.~\ref{fig:3D} depicts evolution of the formation and Fig.~\ref{fig:error2} shows the convergence of $\delta$ to 0.
\begin{figure}[!htb]
	\centering
	\includegraphics[scale = 0.54]{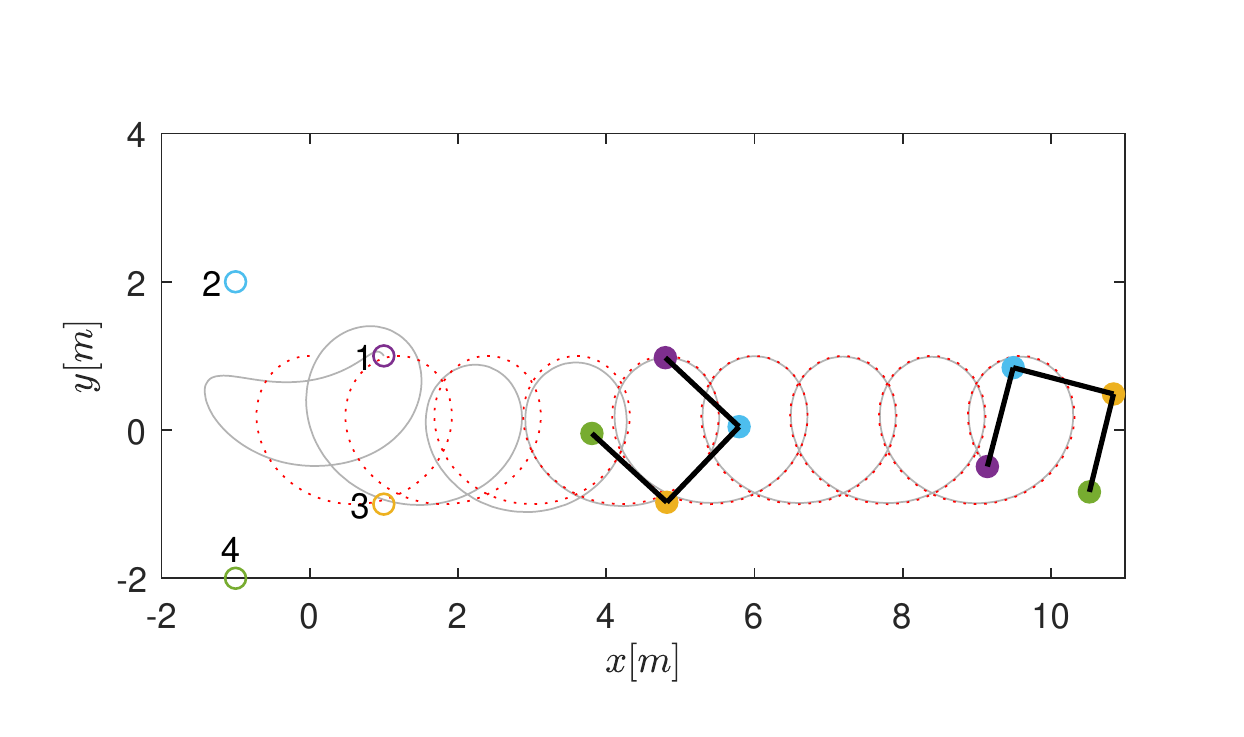}	
	\caption{The figure shows three snapshots: the initial conditions (void circles); $t=50$s, when agents converge to the desired formation; $t=100$s, when agents move along the desired trajectories. The gray solid line represents the trajectory of agent $1$, the dash red line is the desired trajectory of agent $1$ and the black solid lines represent the connections between agents.}
	\label{fig:2D}
\end{figure}
\begin{figure}[!htb]
	\centering
	\includegraphics[scale = 0.44]{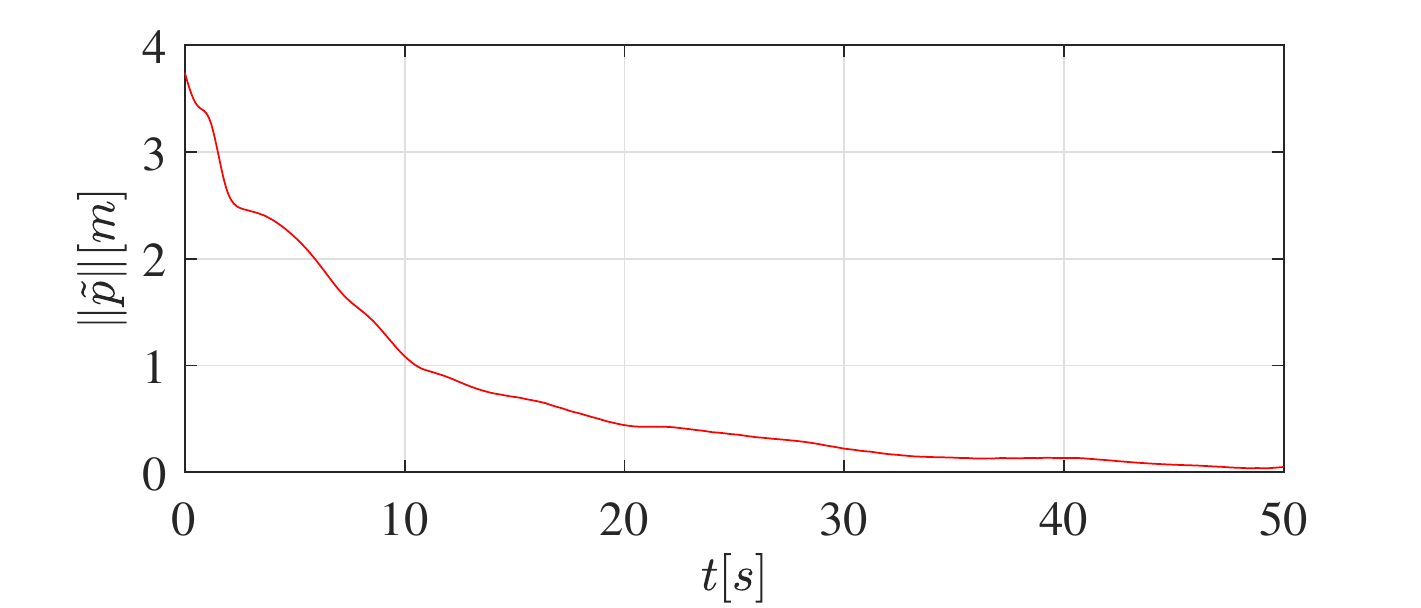}	
	\caption{Time evolution of the norm of the relative position error ($\|\tilde{p}\|$). }
	\label{fig:error1}
\end{figure}
\begin{figure}[!htb]
	\centering
	\includegraphics[scale = 0.54]{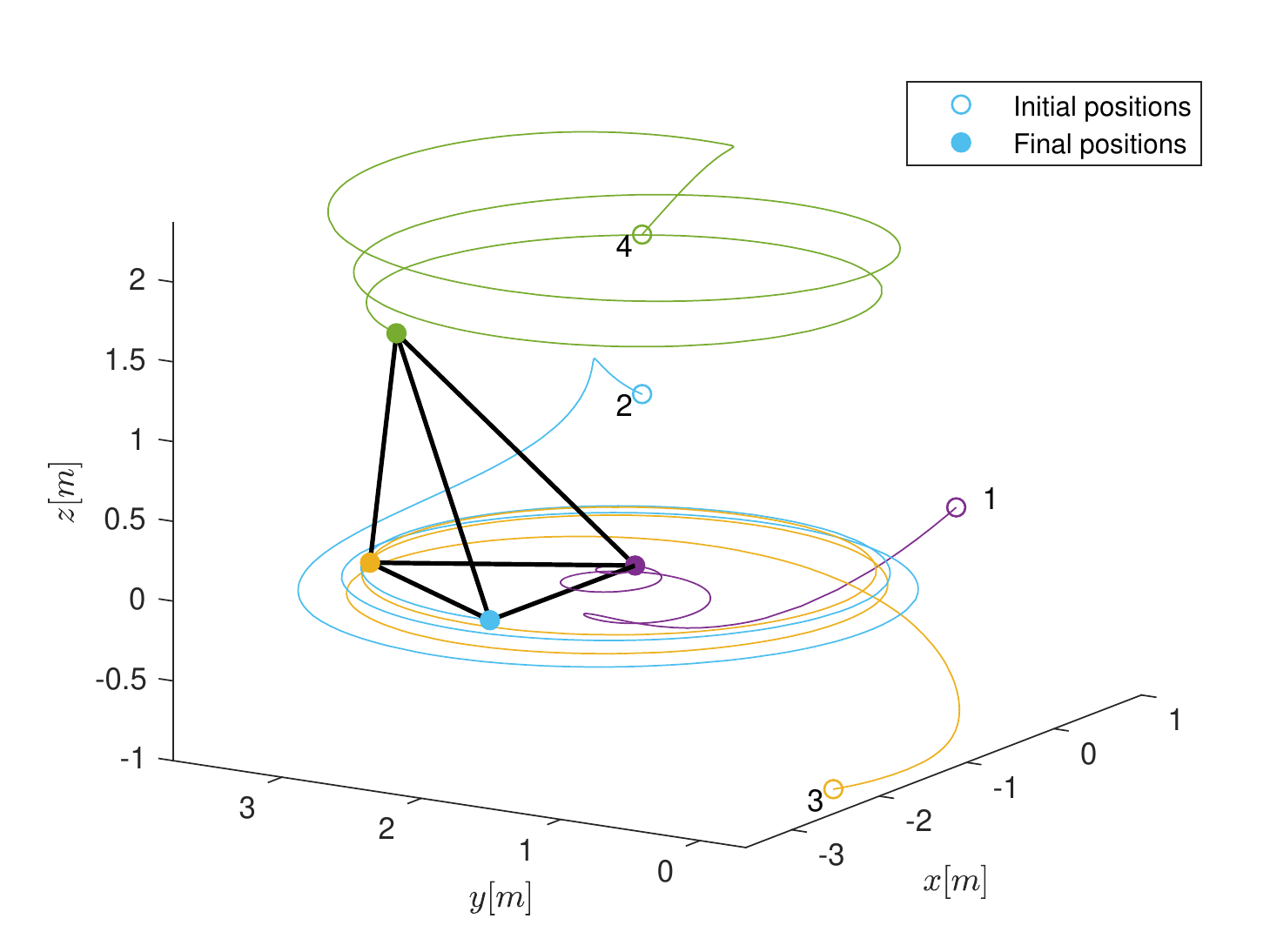}	
	\caption{3-D Trajectory of a pyramid formation. Colorized solid lines represent the agents' trajectory. The black solid lines represent the connections between agents.}
	\label{fig:3D}
\end{figure}
\begin{figure}[!htb]
	\centering
	\includegraphics[scale = 0.44]{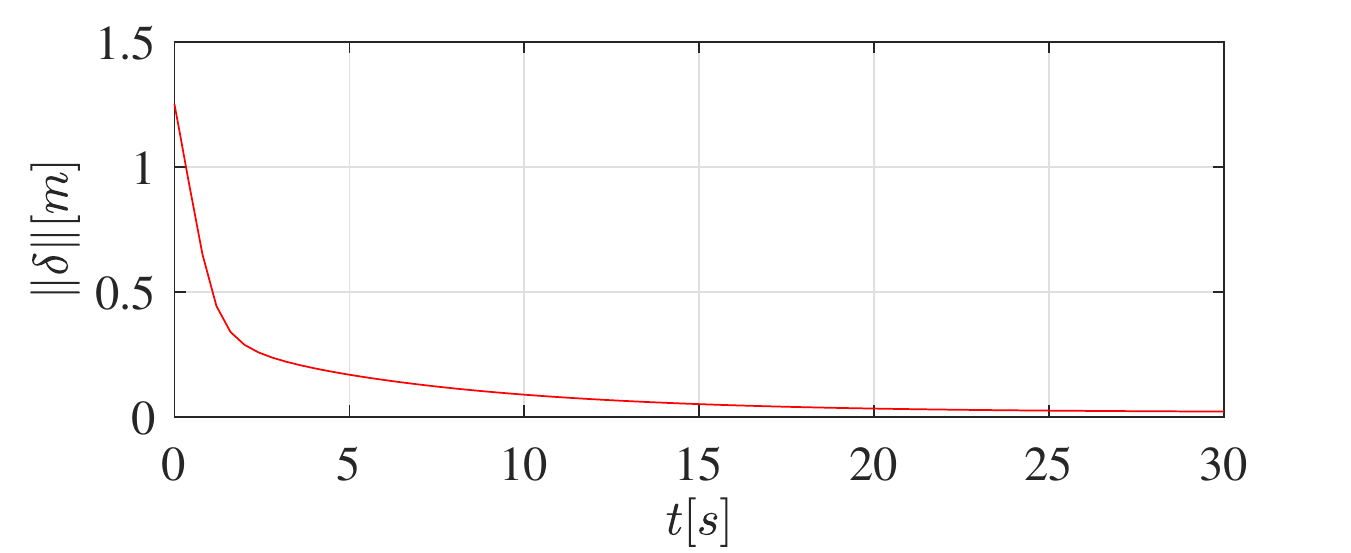}	
	\caption{Time evolution of the norm of the error ($\|\delta\|$). }
	\label{fig:error2}
\end{figure}
\section{Conclusion}
This paper presents new results on formation control of kinematic systems based on time-varying bearing measurements. The key contribution is to show that if the desired formation is bearing persistently exciting,
relaxed conditions on the interaction topology (which do not require bearing rigidity) can be used to derive distributed control laws that guarantee exponential stabilization of the desired formation only up to a translation vector. Simulations results are provided to illustrate the performance of the proposed control method. The future work is to extend these results to account for dynamics (double integrator systems) and include inter-agent collision avoidance.

%

\addtolength{\textheight}{-12cm}   



%
%
%
%
%
%
%

\bibliography{undirected_graph_Final20200828}
\bibliographystyle{IEEEtran}

\end{document}